\documentclass[12pt,envcountsame,envcountsect]{llncs}
\usepackage{geometry}

\usepackage{amsmath,amssymb}
\usepackage{paralist}
\usepackage{graphicx}
\usepackage{xspace}
\usepackage[utf8x]{inputenc}
\usepackage[unicode]{hyperref}

\usepackage{microtype}

\newcommand\TITLE{A Two-Phase Algorithm for Bin Stretching with
  Stretching Factor $\mathbf{1.5}$}

\hypersetup{
  pdftitle={\TITLE},
  pdfauthor={Martin B\"{o}hm, Ji\v{r}\'{\i} Sgall, Rob van Stee, Pavel Vesel\'y},
}

\newcommand\binstretch{{\sc Online Bin Stretching}\xspace}
\newcommand\binpacking{{\sc Bin Packing}\xspace}

\newcommand\eps\varepsilon
\newcommand\calA{{\mathcal A}}
\newcommand\Cprime{{\mathcal C'}}
\newcommand\calB{{\mathcal B}}
\newcommand\calC{{\mathcal C}}
\newcommand\calD{{\mathcal D}}
\newcommand\calE{{\mathcal E}}
\newcommand\calF{{\mathcal F}}
\newcommand\calH{{\mathcal H}}

\newcommand\calL{{\mathcal L}}
\newcommand\calR{{\mathcal R}}
\newcommand\Aprime{{\mathcal A'}}
\newcommand{\hmodfour}{{h}}

\newcommand{\Rbar}{R_{\mbox{\rm\scriptsize first}}}
\newcommand\mycase[1]{\vspace{0.4ex}\noindent{\bf #1}}
\newcommand\mycasesp[1]{\vspace{1ex}\noindent{\bf #1}}
\newcommand*\samethanks[1][\value{footnote}]{\footnotemark[#1]}

\def\indentskip{\hskip 1.5em}
\def\indentskiptwodigit{\hskip 1em}

\newcommand\algobox[1]{\begin{center}
\fbox{\parbox[c]{0.95\textwidth}{#1}}
\end{center}
}

\begin{document}

\title{\TITLE}

\author{Martin B\"{o}hm\inst{1}\fnmsep\thanks{
Supported by the project
14-10003S of GA \v{C}R and by the GAUK project 548214.
},
Ji\v{r}\'{\i} Sgall\inst{1}\fnmsep\samethanks,
Rob van Stee\inst{2}
\and
Pavel Vesel\'y\inst{1}\fnmsep\samethanks
}

\institute{
Computer Science Institute of Charles University,
Prague, Czech Republic.\newline
\email{\{bohm,sgall,vesely\}@iuuk.mff.cuni.cz}.
\and
Department of Computer Science,
University of Leicester, Leicester, UK.\newline
\email{rob.vanstee@leicester.ac.uk}. 
}

\maketitle

\begin{abstract}

\binstretch is a semi-online variant of bin packing in which the
algorithm has to use the same number of bins as an optimal packing,
but is allowed to slightly overpack the bins. The goal is to minimize
the amount of overpacking, i.e., the maximum size packed into any bin.

We give an algorithm for \binstretch with a stretching factor of $1.5$
for any number of bins. We build on previous algorithms and use a
two-phase approach. 
However, our analysis is technically more complicated and uses
amortization over the bins with the help of two weight functions.


\end{abstract}

\section{Introduction}

The most famous algorithmic problem dealing with online assignment is
arguably {\sc Online Bin Packing}.  In this problem, known since the
1970s, items of size between $0$ and $1$ arrive in a
sequence and the goal is to pack these items into the least number of
unit-sized bins, packing each item as soon as it arrives.

\binstretch, which has been introduced by Azar and Regev in
1998~\cite{azar98,azar01}, deals with a similar online
scenario. Again, items of size between $0$ and $1$ arrive in a
sequence, and the algorithm needs to pack each item as soon as it
arrives, but there are the following differences: (i) The packing
algorithm knows $m$, the number of bins that an optimal offline
algorithm would use, and must also use only at most $m$ bins, and (ii)
the packing algorithm can use bins of capacity $R$ for some $R\geq
1$. The goal is to minimize the stretching factor $R$.

In general, the term ``semi-online'' refers to algorithms that,
compared to online algorithms, have some additional global information
about the instance of the problem or have another advantage. We have
formulated \binstretch as a semi-online bin packing variant, where the
algorithm has the additional information that the optimal number of
bins is $m$ and at the same time the algorithm has the advantage of
using bins of larger capacity.

Taking another view, \binstretch can also be thought of as a
semi-online scheduling problem, in which we schedule jobs arriving one
by one in an online manner on exactly $m$ machines, with the objective
to minimize the makespan, i.e., the length of the resulting
schedule. Here the additional information is that the optimal offline
algorithm could schedule all jobs with makespan $1$. Our task is to
present an algorithm with makespan being at most $R$.

%
%

\smallskip
\noindent 
{\bf History.}  
\binstretch has been proposed by Azar and Regev~\cite{azar98,azar01}.
The original lower bound of $4/3$ for three bins has appeared even
before that, in~\cite{KeKoST97}, for two bins together with a matching
algorithm.  Azar and Regev extended the same lower bound to any number
of bins and gave an online algorithm with a stretching factor $1.625$.

The problem has been revisited recently, with both lower bound
improvements and new efficient algorithms. On the algorithmic side,
Kellerer and Kotov~\cite{kellerer2013} have achieved a stretching
factor $11/7 \approx 1.57$ and Gabay et al.~\cite{gabay2013} have
achieved $26/17 \approx 1.53$. There is still a considerable gap to
the best known general lower bound of $4/3$, shown by a simple argument in the
original paper of~Azar and Regev \cite{azar98,azar01}.

While the algorithms and the lower bound from the last paragraph work
for any $m\geq 2$, \binstretch has also been studied in the special
case of a specific constant number of bins $m$. In some cases better
algorithms and lower bounds are possible. Obviously, for $m=1$, the
problem is trivial, as the only possible algorithm has a stretching
factor $1$. For $m=2$, it is easy to achieve a stretching factor
$4/3$, which is optimal. Thus the first interesting case is $m=3$; the
currently best algorithm given by the authors of this paper
\cite{bsvsv15arxiv} has a stretching factor of $11/8 = 1.375$.

Interestingly, the setting with a small fixed number of bins allows
better lower bounds on the stretching factor.  The lower bounds cannot
be easily translated into a lower bound for a larger $m$; for example,
if we modify the instance by adding new bins and a corresponding
number of items of size 1 (that must use exactly the new bins in the
optimum), the semi-online algorithm still could use the additional
capacity of $R-1$ in the new bins to its advantage.  The paper of
Gabay et al.~\cite{gabay2013lbv2} showed a lower bound of
$19/14\approx 1.357$ for $m=3$ using a computer search. Extending
their methods, the authors of this paper were able to reach a lower
bound of $15/11 = 1.\overline{36}$ for $m=3$ in~\cite{bsvsv15arxiv} as
well as a bound of $19/14$ for $m=4$ and $m=5$.  The preprint
\cite{gabay2013lbv2} was updated in 2015 \cite{gabay2013lbv3} to
include a lower bound of $19/14$ for $m=4$ bins. For $m>5$, the
general lower bound of $4/3$ is still the best currently known.

\smallskip\noindent
{\bf Our contributions.} We present a new algorithm for \binstretch
with a stretching factor of $1.5$.  We build on the two-phase approach
which appeared previously in~\cite{kellerer2013,gabay2013}. In this
approach, the first phase tries to fill some bins close to $R-1$ and
achieve a fixed ratio between these bins and empty bins, while the
second phase uses the bins in blocks of fixed size and analyzes each
block separately. 
To reach $1.5$, we needed to significantly improve the analysis using
amortization techniques (represented by a weight function in our
presentation) to amortize among blocks and bins of different types.

A preliminary version of this work appeared in~\cite{bsvsv14}.

\smallskip\noindent
{\bf Related work.} The NP-hard problem \binpacking was originally
proposed by Ullman~\cite{ullman71} and Johnson~\cite{johnson73} in the
1970s. Since then it has seen major interest and progress, see the
survey of Coffman et al.~\cite{coffman13} for many results on
classical {\sc Bin Packing} and its variants. While our problem can be seen
as a variant of {\sc Bin Packing}, note that the algorithms cannot
open more bins than the optimum and thus general results for
\binpacking do not translate to our setting.

As noted, \binstretch can be formulated as online scheduling on
$m$ identical machines with known optimal makespan. Such algorithms
were studied and are important in designing constant-competitive
algorithms without the additional knowledge, e.g., for scheduling in
the more general model of uniformly related
machines~\cite{AsAFPW97,BeChKa00,EbJaSg09}.

For scheduling, also other types of semi-online algorithms are
studied. Historically first is the study of ordered sequences with
non-decreasing processing times~\cite{Graham69}. Most closely related
is the variant with known sum of all processing times studied
in~\cite{KeKoST97}. If we know the optimal makespan, we can always pad
the instance by small items at the end so that the optimal makespan
remains the same and the sum of processing time equals $m$ times the
optimal makespan. Thus one could expect that these two quantities are
interchangeable. However, when the sum of all processing times is
known, the currently best results are a lower bound of $1.585$ and an
algorithm with ratio $1.6$, both
from~\cite{DBLP:journals/tcs/AlbersH12}. This shows, somewhat
surprisingly, that knowing the actual optimum gives a significantly
bigger advantage to the semi-online algorithm over knowing just the
sum of the processing times. See Pruhs et al.~\cite{PST04} for a
survey of other results on (semi-)online scheduling.


\subsection{Definitions and notation} 

Our main problem, \binstretch, can be
described as follows: 

\noindent
{\bf Input:} an integer $m$ and a sequence
of items $I=i_1, i_2, \ldots$ given online one by one. Each item has
a {\it size} $s(i) \in [0,1]$ and must be packed immediately
and irrevocably.

\noindent
{\bf Parameter:} The {\it stretching factor} $R\geq 1$.

\noindent
{\bf Output:} Partitioning (packing) of $I$ into bins $B_1,\ldots,B_m$
so that $\sum_{i\in B_j}s(i)\leq R$ for all $j=1,\ldots,m$.

\noindent
{\bf Guarantee:} there exists a packing of all items in $I$ into $m$
bins of capacity $1$. 

\noindent
{\bf Goal:} Design an online algorithm with the stretching factor $R$
as small as possible which packs all input sequences satisfying the
guarantee.
\smallskip

For a bin $B$, we define the {\it size of the bin} $s(B) = \sum_{i \in
  B} s(i)$.  Unlike $s(i)$, $s(B)$ can change during the course of the
algorithm, as we pack more and more items into the bin. To easily
differentiate between items, bins and lists of bins, we use lowercase
letters for items ($i$, $b$, $x$), uppercase letters for bins and
other sets of items ($A$, $B$, $X$), and calligraphic letters for
lists of bins ($\calA$, $\calC$, $\calL$).

\section{Algorithm} 
\label{sec:bigm}

We rescale the sizes of items and the capacities of bins so that the
optimal bins have capacity 12 and the bins of the algorithm have
capacity 18.

We follow the general two-phase scheme of recent
results~\cite{kellerer2013,gabay2013} which we sketch now.  In the
first phase of the algorithm we try to fill the bins so that their
size is at most 6, as this leaves space for an arbitrary item in each
bin. Of course, if items larger than 6 arrive, we need to pack them
differently, preferably in bins of size at least 12
(since that is the size of the optimal bins). We
stop the first phase when the number of non-empty bins of size at most
6 is three times the number of empty bins. In the second phase, we
work in blocks consisting of three non-empty bins and one empty
bin. The goal is to show that we are able to fill the bins so that the
average size is at least 12, which guarantees we are able to pack the
total size of $12m$ which is the upper bound on the size of all items.

The limitation of the previous results using this scheme was that the
volume achieved in a typical block of four bins is slightly less than
four times the size of the optimal bin, which then leads to bounds
strictly above $3/2$. This is also the case in our algorithm: A
typical block may have three bins with items of size just above 4 from
the first phase plus one item of size 7 from the second phase, while
the last bin contains two items, each of size 7, from the second
phase---a total of 47 instead of desired $4 \cdot 12$. However, we
notice that such a block contains five items of size 7 which the
optimum cannot fit into four bins. To take an advantage of this, we
cannot analyze each block separately as
in~\cite{kellerer2013,gabay2013}.  Instead, the rough idea of our
improved method is to show that a bin with no item of size more than 6
typically has size at least 13 and amortize among the blocks of
different types. Technically this is done using a weight function $w$
that takes into account both the total size of items and the number of
items larger than 6. This is the main new technical idea of our proof.

There are other complications. We would like to ensure that a typical bin
of size at most 6 has size at least 4 after the first phase. However,
this is impossible to 
guarantee if the items packed there have size between 3 and 4. Larger
items are fine, as one per bin is sufficient, and the smaller ones are
fine as well as we can always fit at least two of them.
It is crucial to consider the items with sizes between 3 and 4 very 
carefully.
This motivates
our classification of items: Only the \emph{regular items} of size in
$(0,3]\cup(4,6]$ are packed in the bins filled up to size 6. The
\emph{medium items} of size in $(3,4]$ are packed in their own bins (four or
five per bin). Similarly, \emph{large items} of size in $(6,9]$ are
packed in pairs in their own bins. Finally, the \emph{huge items} of size
larger than $9$ are handled similarly as in the previous papers:
If possible, they are packed with the regular items, otherwise each in
their own bin.

The introduction of medium size items implies that we need to
revisit the analysis of the first phase and also of the case when the
first phase ends with no empty bin. These parts of the proof are
similar to the previous works, but due to the new item type we need to
carefully revisit it; it is now convenient to introduce another 
function $v$ that counts the items according to their type; we call it
a value, to avoid confusion with the main weight function $w$. The
analysis of the second phase when empty bins are present is more
complicated, as we need to take care of various degenerate cases, and
it is also here where the novel amortization is used.

\vspace{8pt}
Now we are ready to proceed with the formal definitions, statement of
the algorithm, and the proof our main result.

\begin{theorem}\label{thm:onepointfive}
There exists an algorithm for \binstretch with a stretching factor of
$1.5$ for an arbitrary number of bins.
\end{theorem}

We take an instance with an optimal packing into $m$ bins of size at
most 12 and, assuming that our algorithm fails, we derive a
contradiction. One way to get a contradiction is to show that the size
of all items is larger than $12m$. As stated above, 
we also use two bounds in the
spirit of weight functions: weight $w(i)$ and value $v(i)$. The weight
$w(i)$ is a slightly modified size to account for items of size larger
than 6. The value $v(i)$ only counts the number of items with
relatively large sizes. For our calculations, it is convenient to
normalize the weight $w(i)$ and value $v(i)$ so that they are at most
0 for bins in the optimal packing (see Lemma~\ref{l:weight}). To get a contradiction, it is then sufficient to prove that
the total weight or value of all bins is positive.

We classify the items based on their size $s(i)$ and define their
value $v(i)$ as follows.  
$$
\begin{array}{c|cccc}
s(i) & (9,12] &(6,9] &(3,4] &(0,3]\cup(4,6]\\
\hline\\[-0.9em]
\mbox{~~type~~} & \mbox{~~huge~~} &\mbox{~~large~~} &\mbox{~~medium~~} &\mbox{~~regular~~}\\
v(i) & 3 & 2 & 1 & 0
\end{array}
$$

\begin{definition}
For a set of items $A$, we define the value $v(A)=(\sum_{i\in
  A}v(i))-3$.

Furthermore we define weight $w(A)$ as follows. Let $k(A)$ be the
number of large and huge items in $A$. Then $w(A)=s(A)+k(A)-13$.

For a set of bins $\calA$ we define $v(\calA)=\sum_{A\in\calA}v(A)$,
$w(\calA)=\sum_{A\in\calA}w(A)$ and $k(\calA)=\sum_{A\in\calA}k(A)$.
\end{definition}

\begin{lemma}
\label{l:weight}
For any packing $\calA$ of a valid input instance into $m$ bins of an
arbitrary capacity,
we have $w(\calA)\leq 0$ and $v(\calA)\leq 0$.
\end{lemma}
\begin{proof}
For the value $v(\calA)$, no optimal bin can contain items
with the sum of their values larger than 3. The bound follows by
summing over all bins and the fact that the number of bins is the same
for the optimum and the considered packing. 

For $w(\calA)$, we have $s(\calA)\leq 12m$ and $k(\calA)\leq m$, as the
optimum packs all items in $m$ bins of volume $12$ and no bin can
contain two items larger than $6$. Thus
$w(\calA)=s(\calA)+k(\calA)-13m\leq 12m+m-13m=0$.
\qed
\end{proof}

\begin{figure}[th]
\begin{center}
\includegraphics[width=0.7\textwidth]{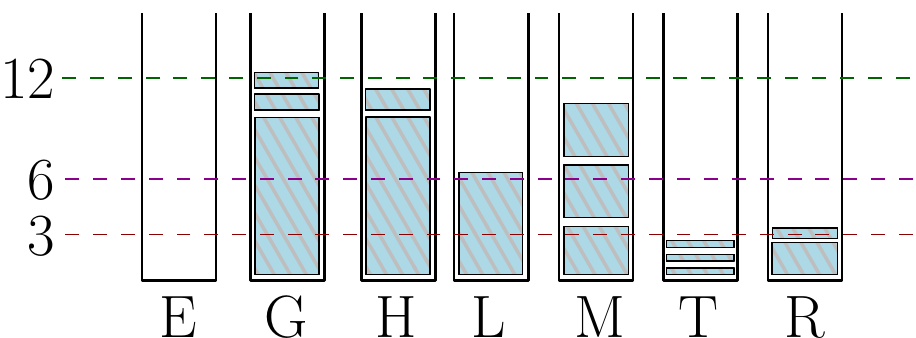}
\end{center}
\caption{An illustration of bin types during the first phase.}
\label{fig:1a}
\end{figure}
\noindent{\bf First phase.}
During the first phase, our algorithm maintains the invariant that only bins
of the following types exist. See Figure~\ref{fig:1a} for an
illustration of the bin types. 

\begin{definition}\label{d:types}
Given a bin $A$, we define the following bin types and introduce
letters that typically denote those bins:

\begin{compactitem}
\item 
{\bf Empty bins (E):} bins that have no item.
\item 
{\bf Complete bins (G):} all bins that have $w(A)\geq 0$ and $s(A)\geq
12$;
\item 
{\bf Huge-item bins (H):} all bins that contain a huge item (plus possibly some other items) and have $s(A)<12$;
\item 
{\bf One large-item bin (L):} a bin containing only a single large
item;
\item 
{\bf One medium-item bin (M):} a non-empty bin with $s(A)<13$ and only medium
items;
\item 
{\bf One tiny bin (T):} a non-empty bin with $s(A)\leq3$;
\item 
{\bf Regular bins (R):} all other bins with $s(A)\in(3,6]$;
\end{compactitem}
\end{definition}

\algobox{{\bf First-phase algorithm:}

During the algorithm, let $e$ be the current number of empty bins and
$r$ the current number of regular bins. 
\medskip
\begin{compactenum}[(1)]
\item While $r < 3e$, consider the next item $i$ and pack it as
  follows, using bins of capacity 18; 
if more bins satisfy a condition, choose among them arbitrarily:
\item \indentskip If $i$ is regular:
\item \indentskip \indentskip If there is a huge-item bin, pack $i$ there.
\item \indentskip \indentskip Else, if there is a regular bin $A$ 
with $s(A)+s(i)\leq 6$, pack it there.
\item \indentskip \indentskip Else, if there is a tiny bin $A$
with $s(A)+s(i)\leq 6$, pack it there.
\item \indentskip If $i$ is medium and there is a medium-item bin where $i$ fits, pack it there.
\item \indentskip If $i$ is large and there is a large-item bin where $i$ fits, pack it there.
\item \indentskip If $i$ is huge:
\item \indentskip \indentskip If there is a regular bin, pack $i$ there.
\item \indentskip \indentskiptwodigit Else, if there is a tiny bin, pack $i$ there.
\item \indentskiptwodigit If $i$ is still not packed, pack it in an
  empty bin.
\end{compactenum}
}

First we observe that the algorithm described in the box above is
properly defined.  The stopping condition guarantees that the
algorithm stops when no empty bin is available. Thus an empty bin is
always available and each item $i$ is packed.  We now state the basic
properties of the algorithm.

\begin{lemma}
\label{l:1} 
At any time during the first phase the following holds:
\begin{compactenum}[\rm(i)]
\item 
\label{i1:classify} 
All bins used by the algorithm are of the types from Definition~\ref{d:types}.
\item 
\label{i1:complete} 
All complete bins $B$ have $v(B)\geq0$.
\item 
\label{i1:hr} 
If there is a huge-item bin, then there is no regular and no tiny bin.
\item 
\label{i1:lm} 
There is at most one large-item bin and at most one medium-item bin.
\item 
\label{i1:tiny} 
There is at most one tiny bin $T$. If $T$ exists, then for any regular
bin, $s(T)+s(R)>6$. There is at most one regular bin $R$ with
$s(R)\leq 4$.
\item 
\label{i1:er} 
At the end of the first phase $3e\leq r \leq 3e+3$.
\end{compactenum}
\end{lemma}
\begin{proof}
(\ref{i1:classify})-(\ref{i1:tiny}): We verify that these invariants
  are preserved when an item of each type arrives and also that the
  resulting bin is of the required type; the second part is always
  trivial when packing in an empty bin.

If a huge item arrives and a regular bin exists, it always fits there,
thus no huge-item bin is created and (\ref{i1:hr}) cannot become
violated. Furthermore, the resulting size is more than $12$, thus the
resulting bin is complete. Otherwise, if a tiny bin exists, the huge
item fits there and the resulting bin is either complete or huge. In
either case, if the bin is complete, its value is 0 as it contains a
huge item.

If a large item arrives, it always fits in a large-item bin if it
exists and makes it complete; its value is at least 1, as it contains
two large items. Thus a second large-item bin is never
created and (\ref{i1:lm}) is not violated.

If a medium item arrives, it always fits in a medium-item bin if it
exists; the bin is then complete if it has size at least 13 and then
its value is at least 1, as it contains 4 or 5 medium items; otherwise
the bin type is unchanged. Again, a second medium-item bin is never
created and (\ref{i1:lm}) is not violated.

If a regular item arrives and a huge-item bin exists, it always fits
there, thus no regular bin is created and (\ref{i1:hr}) cannot become
violated. Furthermore, if the resulting size is at least $12$, the bin
becomes complete and its value is 0 as it contains a huge item; otherwise
the bin type is unchanged.

In the last case, a regular item arrives and no huge-item bin exists.
The algorithm guarantees that the resulting bin has size at most 6,
thus it is regular or tiny. We now proceed to verify
(\ref{i1:tiny}). A new tiny bin $T$ can be created only by packing an
item of size at most 3 in an empty bin. First, this implies that no
other tiny bin exists, as the item would be put there, thus there is
always at most one tiny bin. Second, as the item is not put in any
existing regular bin $R$, we have $s(R)+s(T)>6$ and this also holds
later when more items are packed into any of these bins. A new regular
bin $R$ with $s(R)\leq 4$ can be created only from a tiny bin; note
that a bin created from an empty bin by a regular item is either tiny
or has size in $(4,6]$. If another regular bin with size at most 4
  already exists, then both the size of the tiny bin and the size of
  the new item are larger than 2 and thus the new regular bin has size
  more than 4. This completes the proof of (\ref{i1:tiny}). 

(\ref{i1:er}): Before an item is packed, the value $3e-r$ is at least
  $1$ by the stopping condition. Packing an item may change $e$ or $r$ (or both) by at most $1$. Thus after
  packing an item we have $3e-r\geq 1-3-1=-3$, i.e., $r\leq 3e+3$. If
  in addition $3e\leq r$, the algorithm stops in the next step and
  (\ref{i1:er}) holds.  \qed
\end{proof}

If the algorithm packs all items in the first phase, it
stops. Otherwise according to Lemma~\ref{l:1}(\ref{i1:hr}) we split
the algorithm in two very different branches. If there is at least one
huge-item bin, follow the second phase with huge-item bins below. If
there is no huge-item bin, follow the second phase with regular bins.

Any bin that is complete is not used in the second phase. In
addition to complete bins and either huge-item bins, or regular and
empty bins, there may exist at most three {\em special bins} denoted
and ordered as follows: the large-item bin $L$, the medium-item bin
$M$, and the tiny bin $T$.

\medskip
\noindent
{\bf Second phase with huge-item bins.} In this case, we assume
that a huge-item bin exists when the first phase ends. By Lemma~\ref{l:1}(\ref{i1:hr}),
we know that no regular and tiny bins exist. There are no empty bins either,
as we end the first phase with $3e \le r = 0$. With only a few types of bins
remaining, the algorithm for this phase is very simple:

\algobox{
{\bf Algorithm for the second phase with huge-item bins:}

Let the list of bins $\calL$ contain first all the huge-item bins,
followed by the special bins $L$, $M$, in this order, if they
exist.

\begin{compactenum}[(1)]
\item For any incoming item $i$:
\item \indentskip Pack $i$ using First Fit on the list $\calL$, with
  all bins of capacity 18. 
\end{compactenum}
}

Suppose that we have an instance that has a packing into bins of
capacity 12 and on which our algorithm fails. 
We may assume that the algorithm fails on the last item $f$.  By
considering the total volume, there always exists a bin with size at
most $12$. Thus $s(f)>6$ and $v(f)\ge2$.

If during the second phase an item $n$ with $s(n)\leq6$ is packed into
the last bin in $\calL$, we know that all other bins have size more
than $12$, thus all the remaining items fit into the last bin.
Otherwise we consider $v(\calL)$. Any complete bin $B$ has
$v(B)\geq0$ by Lemma~\ref{l:1}(\ref{i1:complete}) and each huge-item
bin gets nonnegative value, too. Also $v(L)\geq-1$ if $L$
exists. This shows that $M$ must exist, since otherwise
$v(\calL)+v(f)\geq-1+2\geq 1$, a contradiction.

Now we know that $M$ exists, furthermore it is the last bin and thus
we also know that no regular item is packed in $M$. Therefore $M$
contains only medium items from the first phase and possibly large
and/or huge items from the second phase. We claim that
$v(M)+v(f)\geq2$ using the fact that $f$ does not fit into $M$ and $M$
contains no item $a$ with $v(a)=0$: If $f$ is huge we have $s(M)>6$,
thus $M$ must contain either two medium items or at least one medium
item together with one large or huge item and $v(M)\geq -1$.  If $f$
is large, we have $s(M)>9$; thus $M$ contains either three medium
items or one medium and one large or huge item and $v(M)\geq 0$. Thus
we always have $v(\calL)\geq -1+v(M)+v(f)\geq 1$, a contradiction.

\medskip
\noindent
{\bf Second phase with regular bins.}  
Let $\calE$ resp. $\calR$ be the set of empty resp. regular bins at
the beginning of the second phase, and let $e=|\calE|$. Let
$\lambda\in\{0,1,2,3\}$ be such that $|\calR|=3e+\lambda$;
Lemma~\ref{l:1}(\ref{i1:er}) implies that $\lambda$ exists. Note
that it is possible that $\calR=\emptyset$, in that case $e=\lambda=0$.

We organize the remaining non-complete bins into blocks $\calB_i$, and
then order them into a list $\calL$, as follows:

\begin{definition}\label{d:blocks}
Denote the empty bins $E_1, E_2, \dots, E_e$.  The regular
bins are denoted by $R_{i,j}$, $i=1,...,e+1$, $j=1,2,3$. The $i$th
\textbf{block} $\calB_i$ consists of bins $R_{i,1},R_{i,2},R_{i,3},E_i$ in this
order. There are several modifications to this rule:
\begin{compactenum}[\rm (1)] 
\item The first block $\calB_1$ contains
only $\lambda$ regular bins, i.e., it contains $R_{1,1},\ldots,R_{1,\lambda},E_1$ in this
order; in particular, if $\lambda=0$ then $\calB_1$ contains only
$E_1$.

\item The last block $\calB_{e+1}$ has no empty bin,
only exactly $3$ regular bins. 

\item If $e=0$ and $r=\lambda>0$ we define only a single block $\calB_1$
which contains $r=\lambda$ regular bins $R_{1,1},\ldots,R_{1,\lambda}$. 

\item If $e=r=\lambda=0$, there is no block, as there are no empty and
  regular bins. 

\item  If $r>0$, we choose as the first regular bin the one with
size at most $4$, if there is such a bin. 
\end{compactenum}
Denote the first regular bin by $\Rbar$.  If no regular bin exists
(i.e., if $r=0$), $\Rbar$ is undefined.
\end{definition}
Note that $\Rbar$ is either the first bin $R_{1,1}$ in $\calB_1$ if
$\lambda>0$ or the first bin $R_{2,1}$ in $\calB_2$ if $\lambda=0$. By
Lemma~\ref{l:1}(\ref{i1:tiny}) there exists at most one regular bin
with size at most $4$, thus all the remaining $R_{i,j}\neq\Rbar$ have
$s(R_{i,j})>4$.
\begin{definition}\label{d:list}
The \textbf{list of bins} $\calL$ we use in the second phase contains first the
special bins and then all the blocks $\calB_1$, \ldots, $\calB_{e+1}$.
Thus the list $\calL$ is (some or all of the first six bins may not
exist):
\[
L, M, T, R_{1,1}, R_{1,2}, R_{1,3}, E_1, R_{2,1}, R_{2,2}, R_{2,3},
E_2, \dots,
E_e,  R_{e+1,1}, R_{e+1,2}, R_{e+1,3}.
\]
\end{definition}
Whenever we refer to the ordering of the bins, we mean the ordering in
the list $\calL$. See Figure~\ref{fig:1} for an illustration.

\begin{figure}[th]
\begin{center}
\includegraphics[width=\textwidth]{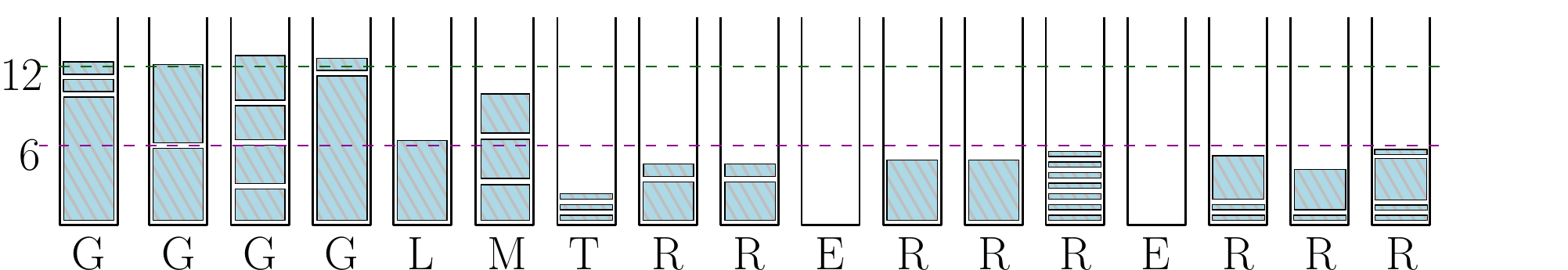}
\end{center}
\caption{A typical state of the algorithm after the first phase. The bin labels correspond to the particular bin types. $G$ denotes complete bins, other labels are the initial letters of the bin types. The non-complete bins (other than $G$) are ordered as in
  the list $\calL$ at the beginning of the second phase with regular
  bins.}
\label{fig:1}
\end{figure}

\algobox{{\bf Algorithm for the second phase with regular bins:}

Let $\calL$ be the list of bins as in Definition~\ref{d:list}, with
all bins of capacity 18. 

\begin{compactenum}[(1)]

\item For any incoming item $i$:
\item \indentskip If $i$ is huge, pack it using First Fit on the
  reverse of the list $\calL$.
\item \indentskip In all other cases, pack $i$ using First Fit on the normal list $\calL$.
\end{compactenum}
}

Suppose that we have an instance that has a packing into bins of
capacity 12 and on which our algorithm fails. We may assume that the
algorithm fails on the last item. Let us denote this item by $f$.
We have $s(f)>6$, as otherwise all bins have size more than 12,
contradicting the existence of optimal packing.
Call the items that arrived in the second phase {\em new} (including
$f$), the items from the first phase are {\em old}.
See Figure~\ref{fig:2} for an illustration of a typical final
situation (and also of notions that we introduce later).

Our overall strategy is to obtain a contradiction by showing that  
$$w(\calL)+w(f)>0\,.$$ 
In some cases, we instead argue that
$v(\calL)+v(f)>0$ or $s(\calL)+s(f)>12|\calL|$. Any of these is
sufficient for a contradiction, as all complete bins have both
value and weight nonnegative and size at least 12.

Let $\calH$ denote all the bins from $\calL$ with a huge item,
and let $\hmodfour=|\calH| \bmod 4$. First
we show that the average size of bins in $\calH$ is large and exclude
some degenerate cases; in particular, we exclude the case when no
regular bin exists.

\begin{lemma}
\label{l:huge} 
Let $\rho$ be the total size of old items
in $\Rbar$ if $\calR\neq\emptyset$ and $\Rbar\in\calH$, otherwise set $\rho=4$.
\begin{compactenum}[\rm (i)]
\item \label{i:hspecial} 
The bins $\calH$ are a final segment of the list $\calL$ and
$\calH\subsetneq\calE\cup\calR$. In particular, $\calR\neq\emptyset$
and $\Rbar$ is defined.
\item \label{i:hr}
We have
$s(\calH)\geq12|\calH|+\hmodfour+\rho-4$. 
\item \label{i:h}
If $\calH$ does not include $\Rbar$, then
$s(\calH)\geq12|\calH|+\hmodfour\geq12|\calH|$. 
\item \label{i:hfirst}
If $\calH$ includes $\Rbar$, then
$s(\calH)\geq12|\calH|+\hmodfour-1\geq12|\calH|-1$.
\end{compactenum} 
\end{lemma}
\begin{proof}
First we make an easy observation used later in the proof. If
$\calE\cup\calR\cup\{T\}$ contains a bin $B$ with no huge item,
then no bin preceding $B$ contains a huge item.  Indeed, if a huge
item does not fit into $B$, then $B$ must contain a new item $i$ of
size at most $9$. This item $i$ was packed using First Fit on the
normal list $\calL$, and therefore it did not fit into any previous
bin. Thus the huge item also does not fit into any previous bin, and
cannot be packed there.

Let $\calH'=\calH\cap(\calE\cup\calR)$. We begin proving our lemma for
$\calH'$ in place of $\calH$.  That is, we ignore the special bins at
this stage. 
The previous observation shows that $\calH'$ is a final segment of the
list.

We now prove the claims (\ref{i:hr})--(\ref{i:hfirst}) with
$\calH'$ in place of $\calH$. 
All bins $R_{i,j}$ with a huge item have size at least $4+9=13$, with
a possible exception of $\Rbar$ which has size at least
$\rho+9=13+\rho-4$, by the definition of $\rho$. Each $E_i$ with a
huge item has size at least $9$. Thus for each $i$ with
$E_i\in\calH'$, $s(E_i)+s(R_{i+1,1})+s(R_{i+1,2})+s(R_{i+1,3})\geq
4\cdot 12$, with a possible exception of $i=1$ in the case when
$\lambda=0$. Summing over all $i$ with $E_i\in\calH'$ and the $\hmodfour$
bins in $\calR$ from the first block intersecting $\calH'$, and
adjusting for $\Rbar$ if $\Rbar\in\calH'$, (\ref{i:hr}) for $\calH'$
follows. The claims (\ref{i:h}) and (\ref{i:hfirst}) for $\calH'$ are
an immediate consequence as $\rho>3$ if $\Rbar\in\calH'$ and $\rho=4$
otherwise.

We claim that the lemma for $\calH$ follows if
$\calH'\subsetneq\calE\cup\calR$. Indeed, following the observation at
the beginning of the proof, the existence of a bin in $\calE\cup\calR$
with no huge item implies that no special bin has a huge item, i.e.,
$\calH'=\calH$, and also $\calH'=\calH$ is a final segment of
$\calL$. Furthermore, the existence of a bin in $\calE\cup\calR$
together with $3e\leq r$ implies that there exist at least one regular
bin, thus also $\Rbar$ is defined and (\ref{i:hspecial}) follows.
Claims (\ref{i:hr}), (\ref{i:h}), and (\ref{i:hfirst}) follow from
$\calH'=\calH$ and the fact that we have proved them for $\calH'$.

Thus it remains to show that $\calH'\subsetneq\calE\cup\calR$.
Suppose for a contradiction that $\calH'=\calE\cup\calR$.

If $T$ exists, let $o$ be the total size of old items in $T$. If also
$\Rbar$ exists, Lemma~\ref{l:1}(\ref{i1:tiny}) implies that
$o+\rho>6>4$, otherwise $o+\rho>4$ trivially. In either case, summing
with (\ref{i:hr}) we obtain
\begin{equation}\label{eq:Told}
  o+s(\calH')>12|\calH'|\,.
\end{equation}

Now we proceed to bound $s(\calH)$. We have already shown claim
(\ref{i:hfirst}) for $\calH'$, i.e., $s(\calH')\geq12|\calH'|-1$. If $L$ or
$M$ has a huge item, the size of the bin is at least 12, as there is
an old large or medium item in it. If $T$ has a huge item, then
(\ref{eq:Told}) implies $s(T)+s(\calH')>9+
o+s(\calH')>9+12|\calH'|$. Summing these bounds we obtain

\begin{equation}
s(\calH) > 12|\calH|-3 \,. \label{eq:T}
\end{equation}
 
We now derive a contradiction in each of the following four cases.

\mycase{Case 1:} All special bins have a huge item.  Then
$\calL=\calH$ and (\ref{eq:T}) together with $s(f)>6$ implies
$s(\calL)+s(f)>12|\calL|+3$, a contradiction.

\mycase{Case 2:} There is one special bin with no huge item.  Then its
size together with $f$ is more than 18, thus (\ref{eq:T}) together
with $s(f)>6$ implies $s(f)+s(\calL)>18+12|\calH|-3>12|\calL|$, a
contradiction.

\mycase{Case 3:} There are two special bins with no huge item and these
bins are $L$ and $M$.

Suppose first that $M$ contains a new item $n$. Then $s(L)+s(n)>18$ by
the First Fit packing rule. The bin $M$ contains at least one old
medium item.  Thus, using (\ref{eq:T}), we get
$s(f)+s(\calL)>6+s(L)+s(n)+3+12|\calH|-3> 24+ 12|\calH|=12|\calL|$, a
contradiction.

If $M$ has no new item, then either $f$ is huge and $M$ has at least
two medium items, or $f$ is large and $M$ has at least three items. In
both cases $v(M)+v(f)\geq 2$. Also $v(L)\geq -1$ since $L$ has a large
item, and $v(\calH)\geq0$ as each bin has a huge item. Altogether we
get that the total value $v(\calL)>0$, a contradiction.

\mycase{Case 4:} There are two or three special bins with no huge
item, one of them is $T$. Observe that the bin $T$ always contains a
new item $n$, as the total size of all old items in it is at most $3$. 

If there are two special bins with no huge item, denote the first one by
$B$. As we observed at the beginning of the proof, if $T$ exists and
has no huge item, no special bin can contain a huge item, thus the
third special bin cannot exist. We have
$s(B)+s(n)>18$, summing with (\ref{eq:Told}) we obtain
$s(f)+s(\calL)\geq
s(f)+s(B)+s(n)+o+s(\calH')>6+18+12|\calH'|=12|\calL|$, a contradiction.

If there are three special bins with no huge item, we have
$s(f)+s(L)>18$ and $s(M)+s(n)>18$. Summing with (\ref{eq:Told}) we
obtain $s(f)+s(\calL)>18+18+o+s(\calH')>36+12|\calH'|=12|\calL|$, a
contradiction.
\qed
\end{proof}

Having proven Lemma~\ref{l:huge}, we can infer existence of the following
two important bins, which (as we will later see) split the instance into three
logical blocks:

\begin{definition} ~ 

\begin{compactitem} \label{d:fc}
\item Let $F$, the \textbf{final bin} be the last bin in $\calL$ before
$\calH$, or the last bin if $\calH=\emptyset$.
\item Let $C$, the \textbf{critical bin}, be the
first bin in $\calL$ of size at most 12. 
\end{compactitem}
\end{definition}

First note that both $F$ and $C$ must exist: $F$ exists by
Lemma~\ref{l:huge}(\ref{i:hspecial}), which also shows that
$F\in\calE\cup\calR$. $C$ exists, as otherwise the total size is more
than $12m$.

To make our calculations easier, we modify the packing so that $f$ is put
into $F$, even though it exceeds the capacity of $F$.
Thus $s(F)>18$ and $f$ (a new item) as well as all the other new items packed
in $F$ or in some bin before $F$ satisfy the property that they do not
fit into any previous bin.  See Figure~\ref{fig:2} for an illustration
of the definitions.

We start by some easy observations. Each bin, possibly with the
exception of $L$ and $M$, contains a new item, as it enters the phase
with size at most 6, and the algorithm failed. Only items of size at
most $9$ are packed in bins before $F$; in $F$ itself only the item
$f$ can be huge.  The bin $F$ always has at least two new items, one that did
fit into it and $f$. All the new items in the bins after $C$ are
large, except for the new huge items in $\calH$ and $f$ which can be
large or huge. (Note that at this point of the proof it is possible
that $C$ is after $F$; we will exclude this possibility soon.)

\begin{figure}[th]
\begin{center}
\includegraphics[width=1\textwidth]{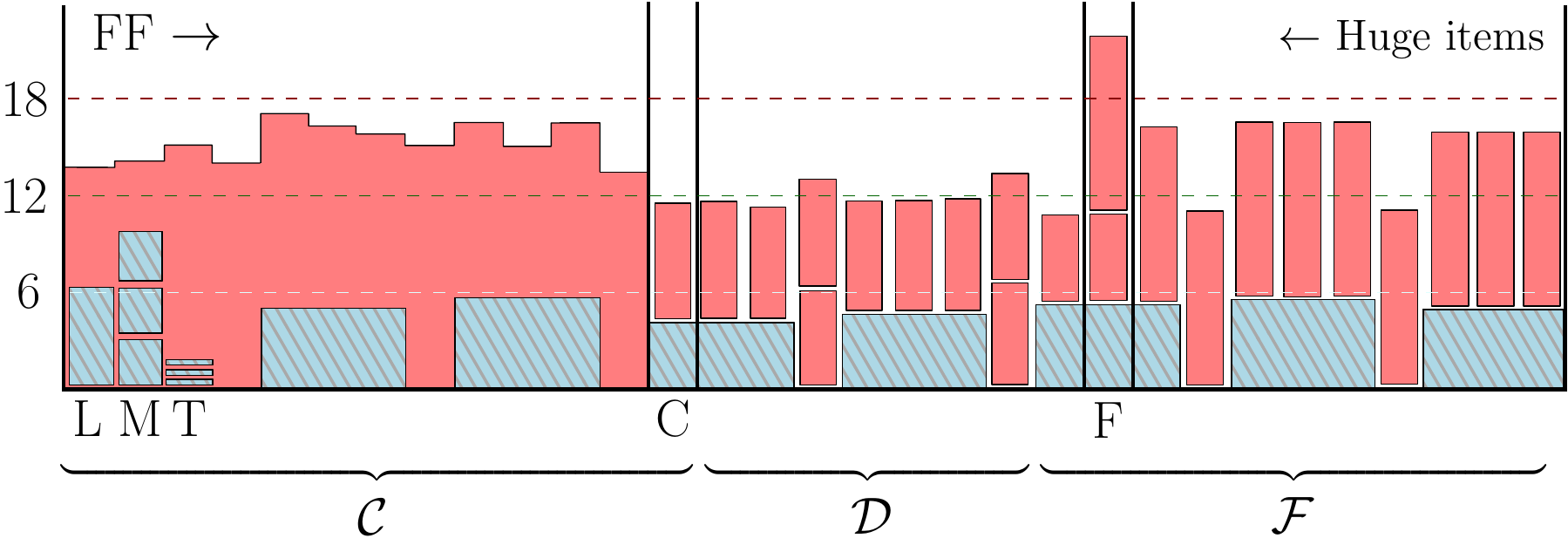}
\end{center}
\caption{A typical state of the algorithm after the second phase with regular
  bins. The gray (hatched) areas denote the old items (i.e., packed in
  the first phase), the red (solid) regions and rectangles denote the
  new items (i.e., packed in the second phase). The bins that are
  complete at the end of the first phase are not shown. The item $f$
  on which the algorithm fails is shown as packed into the final bin
  $F$ and exceeding the capacity $18$, following the convention
  introduced after Definition~\ref{d:fc}.}
\label{fig:2}
\end{figure}

More observations are given in the next two lemmata.

\begin{lemma}
\label{l:aux}
\begin{compactenum}[(i)]
\item
\label{i:9}
Let $B$ be any bin before $F$. Then $s(B)>9$. Furthermore, if
$B\in\calE$ then $B$ contains at least two new items.
\item
\label{i:12}
Let $B,B',B''$ be any three bins in this order before than or equal to
$F$ and let $B''$ contain at least two new items. Then
$s(B)+s(B')+s(B'')>36+o$, where $o$ is the size of old items in $B''$.
\item  
\label{i:11}
Let $B$ be arbitrary and let $B'\in\calR$ be an arbitrary
bin after $B$ and before than or equal to $F$ in $\calL$. 

If $B'\neq\Rbar$ then $s(B)+s(B')>22$, in
particular $s(B)>11$ or $s(B')>11$.

If $B'=\Rbar$ then $s(B)+s(B')>21$.
\end{compactenum}
\end{lemma}

\begin{proof}
$F$ contains a new item $n$ different from $f$. To prove (\ref{i:9}),
  note that $s(n)\leq 9$, and $n$ does not fit into $B$. It follows
  that if $B\in\calE$, then $B$ must contain at least two new items,
  as only items with size smaller than $9$ are packed before $F$.

To prove (\ref{i:12}), let $n,n'$ be two new items in $B''$ and note
that $s(B)+s(n)>18$ and $s(B')+s(n')>18$.

To prove (\ref{i:11}), observe that $B'$ has a new
item of size larger than $18-s(B)$, and it also has old items of size
at least $3$ or even $4$ if $B'\neq\Rbar$.
\qed
\end{proof}

\begin{lemma}
\label{l:cf} 
The critical bin $C$ is before $F$, there are at least two bins
between $C$ and $F$ and $C$ is not in the same block as $F$.
\end{lemma}
\begin{proof}
All bins before $C$ have size larger than $12$. Using
Lemma~\ref{l:huge} we have
$$s(F)+s(\calH)>18+12|\calH|-1=12(|\calH|+1)+5\,.$$ 
It remains to bound the sizes of the other bins. Note that $F\neq C$ as $s(F)>18$.
If $C$ is after $F$,
all bins before $F$ have size more than $12$, so all together
$s(\calL)>12|\calL|+5$, a contradiction. If $C$ is just before $F$,
then by Lemma~\ref{l:aux}(\ref{i:9}),
$s(C)>9=12-3$ and the total size of bins in $s(\calL)>
12|\calL|+5-3>12|\calL|$, a contradiction.

If there is a single bin $B$ between $C$ and $F$, then $s(C)+s(B)$
plus the size of two new items in $F$ is more than $36$ by
Lemma~\ref{l:aux}(\ref{i:12}). If $F\in\calE$ then $\calH$ starts with
three bins in $\calR$, thus $s(\calH)\geq12|\calH|+2$ using 
Lemma~\ref{l:huge} with $\hmodfour=3$, and we get a
contradiction. If $F\in\calR$ then $\Rbar\notin\calH$, thus
Lemma~\ref{l:huge} gives $s(\calH)\geq12|\calH|$, and we get a
contradiction as well.

The last case is when $C$ and $F$ are in the same block with two bins
between them. Then $F\in\calE$, so $\hmodfour=3$, and $C$ is the first
bin of the three other 
bins from the same block, so $\Rbar\notin\calH$. Then $s(C)>9$, the
remaining two bins together with $F$ have size more than $36$ by
Lemma~\ref{l:aux}(\ref{i:12}) and we use $s(\calH)\ge 12|\calH|+3$ from
Lemma~\ref{l:huge} to get a contradiction.
\qed
\end{proof}
We now partition $\calL$ into several parts (see Figure~\ref{fig:2} for an illustration):

\begin{definition}~
\begin{compactitem}
\item Let $\calF=\calB_i\cup\calH$, where $F\in\calB_i$.
\item Let $\calD$ be the set of all bins after $C$ and before $\calF$.
\item Let $\calC$ be the set of all bins before and including $C$.
\end{compactitem}
\end{definition}

Lemma~\ref{l:cf} shows that the
parts are non-overlapping. We analyze the weight of the parts
separately, essentially block by block.
Recall that a weight of a bin is defined as $w(A)=s(A)+k(A)-13$,
where $k(A)$ is the number of large and huge items packed in $A$.
The proof is relatively
straightforward if $C$ is not special (and thus also
$F\not\in\calB_1$), 
which is the most important case driving our choices for $w$. 
A typical block has nonnegative weight, we gain more
weight in the block of $F$ which exactly compensates the loss of
weight in $\calC$, which occurs mainly in $C$ itself. 

Let us formalize and prove the intuition stated in the previous
paragraph in a series of three lemmata.

\begin{lemma}
\label{l:f} 
If $F$ is not in the first block then $w(\calF)>5$, otherwise $w(\calF)>4$.
\end{lemma}
\begin{proof}
All the new items in bins of $\calF$ are large or huge. Each bin has a
new item and the bin $F$ has two new items. Thus $k(\calF)\geq
|\calF|+1$. All that remains is to show that $s(\calF)>12|\calF|+3$,
and $s(\calF)>12|\calF|+4$ if $F$ is not in the first block.

If $F$ is the first bin in a block, the lemma follows as $s(F)>18$ and
$s(\calH)\geq12|\calH|-1$, thus $s(\calF)=s(F)+s(\calH)>12|\calF|+5$.

In the remaining cases there is a bin in $\calR\cap\calF$ before $F$.
Lemma~\ref{l:huge} gives $s(\calH)\geq12|\calH|$; moreover, if
$F\in\calE$, then $s(\calH)\geq12|\calH|+3$.

If $F$ is preceded by three bins from $\calF\cap\calR$, then
$F\in\calE$ and thus $s(\calH)\geq12|\calH|+3$. Using
Lemma~\ref{l:aux}(\ref{i:11}) twice, two of the bins in
$\calF\cap\calR$ before $F$
have size at least $11$ and using Lemma~\ref{l:aux}(\ref{i:9}) the
remaining one has size $9$. Thus the size of these four bins is more
than $11+11+9+18=4\cdot12+1$, summing with the bound for $\calH$ we get
$s(\calF)>12|\calF|+4$. 

If $F$ is preceded by two bins from $\calF\cap\calR$, then by
Lemma~\ref{l:aux}(\ref{i:9}) the total size of these two bins and two
new items in $F$ is more than $36$. If $F\in\calR$, the size of old
items in $F$ is at least $4$ and with $s(\calH)\geq12|\calH|$ we get
$s(\calF)>12|\calF|+4$. 
If $F\in\calE$, which also implies that $F$ is in the first block, then
$s(\calH)\geq12|\calH|+3$, thus $s(\calF)>12|\calF|+3$. 

If $F$ is preceded by one bin $R$ from $\calF\cap\calR$, then let $n$ be a
new item in $F$ different from $f$. We have $s(R)+s(n)>18$ and
$s(f)>6$. We conclude the proof as in the previous case.
\qed
\end{proof}

\begin{lemma}
\label{l:c} ~

If $C\in\calR$ then $w(\calC)\geq -6$. 

If $C\in\calE$ then $w(\calC)\geq -5$. 

If $C$ is a special bin then $w(\calC)\geq -4$.
\end{lemma}

\begin{proof}
For every bin $B$ before $C$, $s(B)>12$ and thus $w(B)>-1$ by the
definition of $C$. Let $\Cprime$ be the set of all bins $B$ before $C$ with
$w(B)\leq0$. This implies that for $B\in\Cprime$, $s(B)\in(12,13]$
and $B$ has no large item. It follows that any new item in any
bin after the first bin in $\Cprime$ has size more than $5$.
We have 
\begin{equation}
\label{eq:wcalC}
w(\calC)\geq w(\Cprime)+w(C)\geq-|\Cprime|+w(C)\,.
\end{equation} 

First we argue that either $|\Cprime|\leq 1$ or $\Cprime=\{M,T\}$.
Suppose that $|\Cprime|>1$, choose $B,B'\in\Cprime$ so that $B$ is before
$B'$.  If $B'\in\calE$, either $B'$ has at most two (new) non-large items and
$s(B')\leq 6+6= 12$, or it has at least three items and
$s(B')>5+5+5=15$; both options are impossible for $B'\in\Cprime$.
If $B'\in\calR$, it has old items of total size in
$(3,6]$. Either $B'$ has a single new item and $s(B')\leq 6+6= 12$, or
it has at least two new items and $s(B')>3+5+5=13$; both options
are impossible for $B'\in \Cprime$. 
The only remaining
option is that $B'$ is a special bin. Since $L$ has a large item,
$L\not\in\Cprime$ and $\Cprime=\{M,T\}$.

By Lemma~\ref{l:aux}(\ref{i:9}), we have $w(C)\geq -4$. The lemma
follows by summing with (\ref{eq:wcalC}) in the following three cases: (i)
$C\in\calR$, (ii) if $\Cprime=\emptyset$ and also (iii) if
both $C\in\calE$ and $|\Cprime|=1$.

For the remaining cases, (\ref{eq:wcalC}) implies that it is
sufficient to show $w(C)\geq -3$. 
If $C\in\calE$ and $\Cprime=\{M,T\}$ then $C$ contains two new items of
size at least 5, thus $w(C)\geq -3$.
If $C=T$ and $\Cprime=\{M\}$ then $C$ either has a large item, or it
has two new items: otherwise it would have size at most 3 of old items
plus at most 6 from a single new item, total of at most 9,
contradicting Lemma~\ref{l:aux}(\ref{i:9}). Thus $w(C)\geq -3$ in this
case as well.
\qed
\end{proof}

\begin{lemma}
\label{l:d} 
\begin{compactenum}[\rm(i)]
\item 
For every block $\calB_i\subseteq\calD$ we have $w(\calB_i)\geq 0$.
\item 
If there is no special bin in $\calD$, then $w(\calD)\geq 0$.
If also $C\in\calR$ then $w(\calD)\geq 1$.
\end{compactenum}
\end{lemma}

\begin{proof}
First we claim that for each block $\calB_i\subseteq\calD$ with three
bins in $\calR$, we have 
\begin{equation}
\label{eq:wB0}
w(\calB_i)\geq 0\,.
\end{equation} 
By Lemma~\ref{l:aux}(\ref{i:11}), one of the bins in $\calR\cap\calB_i$
has size at least 11. By Lemma~\ref{l:aux}(\ref{i:12}), the remaining
three bins have size at least 36. We get (\ref{eq:wB0}) by observing that
$k(\calB_i)\geq 5$, as all the new items placed after $C$ and
before $F$ are large, each bin
contains a new item and $E_i$ contains two new items.

Next, we consider an incomplete block, that is,
a set of bins $\calB$ with at most two bins
from $\calR\cap\calD$ followed by a bin $E\in\calE\cap\calD$.
We claim 
\begin{equation}
\label{eq:wcalB}
w(\calB)\geq 1\,. 
\end{equation}
The bin $E$ contains two large items, since it is after $C$. In particular,
$w(E)\geq 1$ and (\ref{eq:wcalB}) follows if $|\calB|=1$. If $|\calB|=2$, the
size of one item from $E$ plus the previous bin is more than 18, the
size of the other item is more than 6, thus $s(\calB)\geq 24$; since
$k(B)\geq 3$, (\ref{eq:wcalB}) follows.  If $|\calB|=3$, by
Lemma~\ref{l:aux}(\ref{i:12}) we have $s(\calB)\geq 36$; $k(\calB)\geq
4$ and (\ref{eq:wcalB}) follows as well.

By definition, $\calD$ ends by a bin in $\calE$ (if nonempty). Thus
the lemma follows by using (\ref{eq:wcalB}) for the incomplete block,
i.e., for $C\in\calR$ or for $\calB_1$ if it does not have three bins
in $\calR$, and adding (\ref{eq:wB0}) for all the remaining blocks.
Note that $C\in\calR$ implies $\calD\not=\emptyset$.
\qed
\end{proof}

We are now ready to derive the final contradiction. 

If $\calD$ does not contain a special bin, we add the
appropriate bounds from Lemmata~\ref{l:c}, \ref{l:d} and~\ref{l:f}.
If $C\in\calR$ then $F$ is not in the first block and
$w(\calL)=w(\calC)+w(\calD)+w(\calF)>-6+1+5=0$.  If $C\in\calE$ then
$F$ is not in the first block and
$w(\calL)=w(\calC)+w(\calD)+w(\calF)>-5+0+5=0$.  If $C$ is the last
special bin then $w(\calL)=w(\calC)+w(\calD)+w(\calF)>-4+0+4=0$. In
all subcases $w(\calL)>0$, a contradiction.

The rest of the proof deals with the remaining case when $\calD$ does
contain a special bin. This implies there are at least two special
bins and $C$ is not the last special bin. Since $T$ is always the last
special bin (if it exists), it must be the case that $C\neq T$ and
thus $C=L$ or $C=M$.  We analyze the special bins together with the
first block, up to $F$ if $F$ belongs to it.  First observe that the
only bin possibly before $C$ is $L$ and in that case $w(L)\ge0$, so
$w(\calC)\geq w(C)$.

Let $A$ denote $F$ if $F\in\calB_1$ or $E_1$ if $F\not\in\calB_1$.  As
$A=F$ or $A\in\calE$, we know that $A$ contains at least two new
items; denote two of these new items by $n$ and $n'$. Since $A$ is
after $C$, we know that both $n$ and $n'$ are large or huge.

Let $\calA$ be the set
containing $C$ and all bins between $C$ and $A$, not including
$A$. Thus $\calA$ contains two or three special bins followed by at
most three bins from $\calR$. We have $k(\calA)\geq|\calA|-1$ as each
bin in $\calA$ contains a large item, with a possible exception of
$C$ (if $C=M$).  Furthermore $k(A)\geq2$. The bound on $k(\calA)$ and
$k(A)$ imply that
\begin{equation}
\label{eq:ws}
w(\calA)+w(A)\geq s(\calA)+s(A)-12|\calA|-12   
\end{equation}
and thus it is sufficient to bound $s(\calA)+s(A)$. 

The precise bound we need depends on what bin $A$ is.  In each case,
we first determine a sufficient bound on $s(\calA)+s(A)$ and argue
that it implies contradiction. Afterwards we prove the bound.
Typically, we bound the size by creating pairs of bins of size $21$ or
$22$ by Lemma~\ref{l:aux}(\ref{i:11}). We also use that $s(B)>9$ for
any $B\in\calA$ by Lemma~\ref{l:aux}(\ref{i:9}) and that $n$, $n'$
together with any two bins in $\calA$ have size at least $36$ by
Lemma~\ref{l:aux}(\ref{i:12}).

\mycasesp{Case $A\neq F$:} 
Then $F\not\in\calB_1$ and $A=E_1$. We claim that 
\begin{equation}
\label{eq:AneqF}
s(\calA)+s(A)\geq12|\calA|+7\,.
\end{equation}
First we show that~(\ref{eq:AneqF}) implies a contradiction. Indeed, 
(\ref{eq:AneqF}) together with (\ref{eq:ws}) yields
$w(\calA)+w(A)\geq -5$ and summing this with all the other bounds,
namely $w(\calF)>5$ from Lemma~\ref{l:f} and $w(\calB_i)\ge0$ for
whole blocks $\calB_i\in\calD$ from Lemma~\ref{l:d}, leads to
$w(\calL)>0$, which is a contradiction.  

Now we prove~(\ref{eq:AneqF}). The items $n$ and $n'$ from $A$
together with the first two special bins in $\calA$ have size more
than 36. Let $\Aprime$ be the set of the remaining bins; it contains
possibly $T$ and at most three bins from $\calR$.  It remains to show
$s(\Aprime)\geq 12|\Aprime|-5$.

For $|\Aprime|=0$ it holds trivially.  

If $|\Aprime|=1$, the only bin in $\Aprime$ has size more than
9 and this is sufficient. 

For $|\Aprime|>1$ we apply Lemma~\ref{l:aux}(\ref{i:11}) and pair as
many bins from $\Aprime$ as possible; note that all the bins in
$\Aprime$ except possibly $T$ are in $\calR$, so the assumptions of
the lemma hold.  If $|\Aprime|=2$, then
$s(\Aprime)>21=2\cdot12-3$. For $|\Aprime|=3$ we get
$s(\Aprime)>22+9=3\cdot12-5$, since we can create a pair without
$\Rbar$.  Finally, if $|\Aprime|=4$ then $s(\Aprime)>22+21=4\cdot
12-5$.

\mycasesp{Case $A=F$:} 
We claim that it is sufficient to prove
\begin{equation}
\label{eq:AF}
s(\calA)+s(n)+s(n')>12|\calA|+\begin{cases}
8 & \mbox{if $F\in\calR$ and $\Rbar\in\calA$,} \\
9 & \mbox{if either $F\in\calR$ or $\Rbar\in\calA$,} \\
10 & \mbox{in all cases.}
\end{cases}
\end{equation}
First we show that~(\ref{eq:AneqF}) implies a contradiction. 
 
If $F=E_1$ we note that $\hmodfour=3$ (as $\calH$ starts with $3$ bins
in $\calR$). Thus Lemma~\ref{l:huge}, items (\ref{i:h}) and
(\ref{i:hfirst}), together with $w(\calH)\geq s(\calH)-12|\calH|$
yields $w(\calH)\geq 3$ for $\Rbar\in\calA$ or $w(\calH)\geq 2$ for
$\Rbar\not\in\calA$. Summing this with $w(\calA)+w(A)>-3$ or
$w(\calA)+w(A)>-2$, that are obtained from (\ref{eq:ws}) and
(\ref{eq:AF}) in the respective cases, we obtain $w(\calL)>0$, a
contradiction.

If $F\in\calR$ then we know that $F$ also contains old items
of size at least 3 if $\Rbar\not\in\calA$ or even 4 if $\Rbar\in\calA$
(and thus $F\neq\Rbar$). Summing this with the respective bound from
(\ref{eq:AF}) we obtain $s(\calA)+s(F)>12|\calA|+12$. Summing this
with $s(\calH)\geq 12|\calH|$ from Lemma~\ref{l:huge}(\ref{i:h}) now
yields $s(\calL)>12|\calL|$, a contradiction.

Thus~(\ref{eq:AneqF}) always leads to a contradiction. 

We now distinguish subcases depending on $|\calA|$ and in each case we
either prove~(\ref{eq:AF}) or obtain a contradiction directly. 
Note that $\Rbar\in\calA$ whenever $|\calA|\geq 4$. 

\mycasesp{Case $|\calA|=2$:} The two bins together with $n$ and
$n'$ have size more than $36$. Thus
$s(\calA)+s(n)+s(n')>36=12\cdot2+12$, which implies~(\ref{eq:AF}). 

\mycasesp{Case $|\calA|=3$:} We have $s(C)>9$ and the remaining two
bins together with $n$ and $n'$ have size more than $36$. Thus
$s(\calA)+s(n)+s(n')>12|\calA|+9$, which implies~(\ref{eq:AF}) in all
cases except if $F=E_1$ and $\Rbar\not\in\calA$.  

In the remaining case, $\calA=\{L,M,T\}$ and $C=L$, as $\calA$
contains no bin from $\calR$ and $|\calA|=3$.  We prove a
contradiction directly. Let $o$ be the size of old items in $T$.  We
apply Lemma~\ref{l:huge}(\ref{i:hr}), using the fact that $o+\rho>6$
by Lemma~\ref{l:1}(\ref{i1:tiny}), where $\rho$ is the total size of
old items in $\Rbar\in\calH$, and $\hmodfour=3$. We get
$o+s(\calH)\geq o+12|\calH|+\hmodfour+\rho-4>12|\calH|+5$.  Let $n''$
be a new item in $T$. Since $n''$ does not fit into $M$,
$s(M)+s(n'')>18$; also $s(L)>9$ and $s(F)>18$. Summing all the bounds,
we have $s(\calL)\geq o+s(\calH)+s(M)+s(n'')+s(L)+s(F)
>12|\calH|+5+18+9+18=12|\calL|+2$, a contradiction.

\mycasesp{Case $|\calA|=4$:}
The last bin $R\in\calA$ is in
$\calR$. Together with any previous bin it has size more than $21$,
the remaining two bins together with $n$ and $n'$ have size more than
$36$ by Lemma~\ref{l:aux}(\ref{i:12}). 
Thus $s(\calA)+s(n)+s(n')>21+36=4\cdot12+9$ which implies~(\ref{eq:AF}),
since $\Rbar\in\calA$.

\mycasesp{Case $|\calA|=5$:}
First consider the case $F=E_1$. 
The last two bins of $\calA$ are in $\calR$, we pair
them with two previous bins to form pairs of size more than
The remaining bin has size at least $9$,
since $n$ does not fit into it and $s(n)<9$. We also have
$s(F)>18$. Thus $s(\calA)+s(A)>21+22+9+18=5\cdot12+10$, which
implies~(\ref{eq:AF}).

If $F\in\calR$ then one of the last two bins of $\calA$ has size more
than $11$ and the other forms a pair of size more than $21$ with one
special bin. The remaining two bins together with $n$ and $n'$ have
size more than 36 by Lemma~\ref{l:aux}(\ref{i:12}). Thus
$s(\calA)+s(n)+s(n')>11+21+36=5\cdot12+8$ which implies~(\ref{eq:AF}),
because $\Rbar\in\calA$.

\mycasesp{Case $|\calA|=6$:} Then $\calA$ contains all three special
bins and three bins from $\calR$, therefore also $F=E_1$. We form
three pairs of a special bin with a bin from $\calR$ of total size
more than $21+22+22$. Since $s(F)>18$, we have
$s(\calA)+s(F)>21+22+22+18=6\cdot12+11$.  Since in this case
$A=F=E_1$, we have $s(\calH)\geq 12|\calH|+3$ and
$s(\calL)>12|\calL|$, a contradiction.

In all of the cases we can derive a contradiction, which implies
that our algorithm cannot fail. This concludes the proof of
Theorem~\ref{thm:onepointfive}. \qed

\section{Conclusions}

We note that the analysis of our algorithm is tight, i.e., if we
reduce the capacity of the bins below 18, the algorithm
fails. Consider the following instance. Send two items of size 6
which are in the first phase packed separately into two bins. Then
send $m-1$ items of size 12. One of them must be put into a bin
with an item of size 6, i.e., one bin receives items of size 18, while
all the items can be packed into $m$ bins of size 12.

To decrease the upper bound below $1.5$ seems challenging. In
particular, the instance above and its modifications with more items
of size 6 or slightly smaller items at the beginning shows that these
items need to be packed in pairs. This in turns creates difficulties
that, in the current approach, lead to new item and bin types; at this
point we do not know if such an approach is feasible.

Another possible research direction is to improve the lower bounds, in
particular for large $m$. It is quite surprising that there are no
lower bounds for $m>5$ larger than the easy bound of $4/3$.

\smallskip

\noindent{\bf Acknowledgment.} The authors thank Emese Bittner for
useful discussions during her visit to Charles University. We also
thank to referees for many useful comments.

\end{document}